\documentclass[smallextended]{svjour3}       
\smartqed  
\usepackage{graphicx}
\usepackage{tikz}
\usetikzlibrary{intersections,calc}

\usepackage{booktabs} 
\usepackage{paralist} 
\usepackage{verbatim} 
\usepackage{subfig} 
\usepackage{textcomp}
\usepackage{amsmath}
\usepackage{bbm}
\usepackage{bbold}
\usepackage{mathtools}
\usepackage{amsfonts}
\usepackage{amssymb}
\usepackage{algorithm}
\usepackage{algpseudocode}
\usepackage[inline]{enumitem}
\usepackage{sgamevar}
\usepackage{color} 
\usepackage{kbordermatrix}
\usepackage{natbib} 
\setcitestyle{aysep={}} 

\newcommand{\ts}{\textsuperscript} 

\newtheorem{fact}{Fact}

\spdefaulttheorem{assumption}{Assumption}{\it}{\rm}

\newcommand{\colspan}[1]{\texttt{ColSpan}(#1)}
\newcommand{\nullspace}[1]{\texttt{null}(#1)}
\newcommand{\rank}[1]{\texttt{rank}\left({#1}\right)}
\newcommand{\m}[1]{\texttt{m}({#1})}

\newcommand{\transpose}{\mathsf{T}}
\renewcommand{\Re}{\mathbb{R}}
\newcommand{\vct}[1]{\boldsymbol{#1}} 
\newcommand{\Rea}{\mathbb{R}}
\newcommand{\Ce}{\mathbb{C}}

\newcommand{\bigO}{\mathcal{O}}
\DeclareMathOperator*{\argmin}{argmin}
\DeclareMathOperator*{\argmax}{argmax}
\newcommand{\mat}[1]{\begin{bmatrix} #1 \end{bmatrix}}
\begin{document}

\title{Rank Reduction in Bimatrix Games \thanks{*The first author was fully supported by the United
States Military Academy and the Army Advanced
Civil Schooling (ACS) program. The second author gratefully acknowledges support from NSF CRII Grant 1565487 and NSF ECCS Grant 1610615.}
}

\author{Joseph L. Heyman \and Abhishek Gupta}

\institute{J.L. Heyman \at
            United States Military Academy Westpoint\\
              606 Thayer Rd, West Point, NY 10996\\
              Tel.: +1-719-494-3027\\
              \email{joseph.heyman@westpoint.edu}           
           \and
           A. Gupta \at
              2015 Neil Avenue, Columbus, OH 43210\\
              \email{gupta.706@osu.edu}
}

\maketitle

\begin{abstract}
The rank of a bimatrix game is defined as the rank of the sum of the payoff matrices of the two players. The rank of a game is known to impact both the most suitable computation methods for determining a solution and the expressive power of the game. Under certain conditions on the payoff matrices, we devise a method that reduces the rank of the game without changing the equilibrium of the game. We leverage matrix pencil theory and the Wedderburn rank reduction formula to arrive at our results. We also present a constructive proof of the fact that in a generic square game, the rank of the game can be reduced by 1, and in generic rectangular game, the rank of the game can be reduced by 2 under certain assumptions.

\end{abstract}

\section{Introduction}\label{sec:intro}
The study of game theory --- which models strategic interactions among rational agents --- has a rich history dating back to the formalization of the field by John \cite{neumann1928theorie}; see \cite{von1959theory} for an English translation. Nash equilibrium (NE) is the widely accepted solution approach for normal form games, and computing a solution in a $k$-player finite game is one of the fundamental problems in game theory. Due to the well known theorem by \cite{nash1951}, we know that every finite game has a solution, possibly in mixed strategies. However, outside of some restricted classes of games, which includes zero-sum games, rank-1 games, potential games, and coordination games, the computation of NE is shown to be computationally difficult in the worst case.

In a two-player bimatrix game in which the row player has $m$ pure strategies and the column player has $n$ pure strategies, the payoffs to the players can be represented as two matrices, $A$ and $B$ in $\Re^{m\times n}$. Define the rank of this game as the rank of the sum of the two payoff matrices, $\rank{A+B}$ \citep{kannan2007games,kannan2010games}. The rank of a game is known to impact both the most suitable computation methods for determining a solution and the expressive power of the game. For example, it is well known that zero-sum games, which are rank-0 games, can be solved via a linear programming approach. For approximate solutions, \cite{kannan2007games,kannan2010games} present an algorithm for computing approximate Nash equilibria of rank-$k$ games, defined as the class of games where $\rank{A+B}\leq k$, for some given $k$.  

The concept of \textit{strategic equivalence} between game theoretic models also enjoys a rich history, dating back to at least von Neumann and Morgenstern's book first published in 1944 \cite[p.~245]{von2007theory}. If two games are \textit{strategically equivalent}, then they have the same set of Nash equilibria. See Subsection \ref{subsec:stratEq} for many competing definitions of strategic equivalence between games. Operations that preserve the strategic equivalence of bimatrix games can modify the rank of the game. For example, the well-studied constant-sum game, in which the sum of the two payoff matrices equals a constant matrix, is strategically equivalent to a zero-sum game. However, the zero-sum game has rank zero, while the constant-sum game is a rank-$1$ game. Since the rank of a game influences both the most suitable solution techniques and the expressive nature of a game, one should be particularly interested in determining if a given game  is strategically equivalent to a lower rank game.

Accordingly, in this paper, given a game, we apply the classical theory of matrix pencils in conjunction with the Wedderburn rank reduction formula to determine whether or not the given game is strategically equivalent to a game of lower rank. If so, we also devise an algorithm for efficiently computing the strategically equivalent lower rank game.

\subsection{Prior Work}\label{subsec:priorWork}
Closely related to our work is the class of \textit{strategically zero-sum} games defined by \cite{moulin1978strategically}. They study the class of games in which no completely mixed NE can be improved upon via a correlation strategy and come to the conclusion that these games are the class of \textit{strategically zero-sum} games. For the bimatrix case, they provide a complete characterization of strategically zero-sum games \cite[Theorem~2]{moulin1978strategically}. \cite{kontogiannis2012mutual} introduce the notion of mutually concave games and show that such games are solvable in polynomial time. They then proceed to characterize this class of games and conclude that this class is precisely the class of strategically zero-sum games. They devise an algorithm that can compute the strategically equivalent zero-sum game in $\bigO{(m^3n^3)}$ steps.

Around the same time, \cite{isaacson1980class} studied a class of bimatrix games that they characterized as \textit{row-constant games} . They define row-constant games as those bimatrix games where the sum of the payoff matrices is a matrix with constant rows. In their work, they show that the NE strategies of a row-constant game can be found via solving the zero-sum game $(m,n,A,-A)$. Comparing \cite{isaacson1980class} and \cite{moulin1978strategically}, one can easily see that \textit{row-constant} games form a subclass of \textit{strategically zero-sum} games.  

Related to strategically zero-sum bimatrix games are the class of \textit{strictly competitive games} \citep{aumann1961almost}. In a strictly competitive game, if both players change their mixed strategies, then either the payoffs remain unchanged, or one of the two payoffs increases while the other payoff decreases. In other words, all possible outcomes are Pareto optimal. It has long been claimed that strictly competitive games share many common and desirable NE features with zero-sum games, such as ordered interchangeability, NE payoff equivalence, and convexity of the NE set. Indeed, Aumann claims that strictly competitive games are equivalent to zero-sum games \citep{aumann1961almost}. Moulin and Vial proceed to cite Aumann's claim when arguing that strictly competitive games form a subclass of strategically zero-sum games \cite[Example~2]{moulin1978strategically}. However, many years later \cite{adler2009note} conducted a literature search and found that the claim of equivalence of strictly competitive games and zero-sum games was often repeated, but without a formal proof. They then proceeded to prove that this claim does indeed hold true.

Another class of games that shares some NE properties with zero-sum games is the class of \textit{(weakly) unilaterally competitive games}. As defined by \cite{kats1992unilaterally}, a game is unilaterally competitive if a unilateral change in strategy by one player results in a (weak) increase in that player's payoff if and only if it results in a (weak) decrease in the payoff of all other players. A game is weakly unilaterally competitive if a unilateral change in strategy by one player results in a strict increase in that player's payoff, then the payoffs of all other players (weakly) decrease.  If the payoff of the player who makes the unilateral move remains unchanged, then the payoffs of all players remain unchanged. For the two-player case, strictly competitive games form a subclass of unilaterally competitive games, and unilaterally competitive games form a subclass of weakly unilaterally competitive games. For the bimatrix case, \cite{kats1992unilaterally} show that (weakly) unilaterally competitive games have the ordered interchangeability and NE payoff equivalence properties. However, convexity of the NE set is only proved for infinite games with quasiconcave payoff functions. 

More recently, in a series of works \cite{adsul2011rank,adsul2020fast} the authors have developed polynomial time algorithms for solving another subclass of bimatrix games called rank-$1$ games. Rank-$1$ games are defined as games where the sum of the two payoff matrices is a rank-$1$ matrix. As far as we are aware, no other authors have charaterized games that are strategically equivalent to rank-$1$ games.

There is a significant body of research into computation of approximate Nash equilibria in bimatrix games; see, for instance, \cite{bosse2007new,daskalakis2006note,barman2018approximating,lipton2003,daskalakis2007progress,kontogiannis2006polynomial,kontogiannis2010well,alon2013approximate}. The approximation bounds of the Nash equilibria computed by these algorithms depend on sparsity of the sum of payoff matrices \cite{barman2018approximating}, rank of the payoff matrices \cite{lipton2003}, or the rank of the sum of payoff matrices \cite{kannan2010games,kannan2007games}. We believe that our algorithm can be used as a preprocessing step to reduce the rank of the game first and then apply the algorithm to compute an approximate Nash equilibrium.

\subsection{Notation}\label{subsec:notation}
We use $\vct{1}_n$ and $\vct{0}_n$ to denote, respectively, the all ones and all zeros vectors of length $n$. All vectors are annotated by bold font, e.g $\vct{u}$, and all vectors are treated as column vectors. $\Delta_n$ is the set of probability distributions over $ \{1,\ldots,n\}$, where $\Delta_{n}=\big\{\mathbf{p} \mid p_i\geq0,\forall i\in \{1,\dots,n\},\sum_{i=1}^n p_i = 1  \big\}$.
Let $\vct{e}_j$, $j \in \{1,2,...,n\}$, denote the vector with $1$ in the $j$\ts{th} position and $0$'s elsewhere. 

Consider a matrix $C$. We use $\rank{C}$ to indicate the rank of the matrix $C$. $\colspan{C}$ indicates the subspace spanned by the columns of the matrix $C$, also known as the column space or the range of the matrix $C$. We indicate the nullspace of the matrix $C$, the space containing all solutions to $Cx=\vct{0}_m$, as $\nullspace{C}$. In addition, we use $C^{(j)}$ to denote the $j$\ts{th} column of $C$ and $C_{(i)}$ to denote the $i$\ts{th} row of $C$. 

\subsection{Outline of the Paper}
In the next section, we review some game theoretic concepts. In Section \ref{sec:ser}, we discuss various notions of strategic equivalence used in the literature and implications of those definitions. In this paper, we will use the definition of strategic equivalence introduced in \cite{moulin1978strategically}. We present the main result in Section \ref{sec:MainResult}. In that section, we first briefly review the Wedderburn rank reduction formula and the theory of matrix pencils. We then introduce the overall rank reduction algorithm in Subsection \ref{sub:MainResult}. We then devise in Section \ref{sec:rank1Pencils} a novel algorithm for determining the solution to the rank-1 matrix pencil problem, assuming such a solution exists. Next, we present some consequences of our results when applied to generic (random) games in Section \ref{sec:generic}. We present a simple numerical example in Section \ref{sec:numerical} and discuss the complexity of our algorithm in this section. We finally conclude the paper in Section \ref{sec:conclusion}.

\section{Preliminaries}\label{sec:prelim}
In this section, we recall some basic definitions in bimatrix games and the definition of strategic equivalence in bimatrix games.

We consider here a two player game, in which player 1 (the row player) has $m$ actions and player 2 (the column player) has $n$ actions. Player 1's set of pure strategies is denoted by $S_1=\{1,\dots,m\}$ and player 2's set of pure strategies is $S_2=\{1,\dots,n\}$. If the players play pure strategies $(i,j)\in S_1 \times S_2$, then player 1 receives a payoff of $a_{ij}$ and player 2 receives $b_{ij}$.

We let $A=[a_{ij}]\in \Rea^{m\times n}$ represent the payoff matrix of player 1 and $B=[b_{ij}]\in \Rea^{m\times n}$ represent the payoff matrix of player 2. As the two-player finite game can be represented by two matrices, this game is commonly referred to as a bimatrix game. The bimatrix game is then defined by the tuple $(m,n,A,B)$. Define the $m\times n$ matrix $C$ as the sum of the two payoff matrices, $C \coloneqq A+B$. We define the rank of a game as $\rank{C}$. Some authors define the rank of the game to be the maximum of the rank of the two matrices $A$ and $B$, but this is not the case here. 

Players may also play mixed strategies, which correspond to a probability distribution over the available set of pure strategies. Player 1 has mixed strategies $\vct{p}$ and player 2 has mixed strategies $\vct{q}$, where $\vct{p}\in \Delta_m$ and $\vct{q}\in \Delta_n$. Using the notation introduced above, player 1 has expected payoff $\vct{p}^\transpose A \vct{q}$ and player 2 has expected payoff $\vct{p}^\transpose B\vct{q}$. 

\subsection{Nash Equilibrium in Bimatrix Games}\label{subsec:stratEq}
To introduce the notion of Nash equilibrium, we first need to define the best response correspondence. For the payoff matrices $A,B$, define the best response correspondences $\Gamma_A:\Delta_n\rightrightarrows \Delta_m$ and $\Gamma_B:\Delta_m\rightrightarrows \Delta_n$ as
\begin{align*}
 \Gamma_A(\vct q) &= \{\vct p\in\Delta_m: \vct p^\transpose A \vct q = \max_i [A\vct q]_i \},\\
 \Gamma_B(\vct p) &= \{\vct q\in\Delta_n: \vct p^\transpose B \vct q = \max_j [\vct p^TB]_j \}.
\end{align*}

A Nash Equilibrium is defined as a tuple of strategies $(\vct{p}^*,\vct{q}^*)$ such that each player's strategy is an optimal response to the other player's strategy. In other words, neither player can benefit, in expectation, by unilaterally deviating from the Nash Equilibrium. This is made precise in the following definition.
\begin{definition}[Nash Equilibrium \citep{nash1951}]\label{def:NE}
	A pair $(\vct{p}^*,\vct{q}^*)$ of mixed strategies is a Nash Equilibrium (NE) if and only if $\vct p^*\in \Gamma_A(\vct q^*)$ and $\vct q^*\in \Gamma_B(\vct p^*)$.
\end{definition}
It is a well known fact due to \cite{nash1951} that every bimatrix game with a finite set of pure strategies has at least one NE, possibly in mixed strategies.

\section{Strategically Equivalent Games and Problem Formulation}\label{sec:ser}
Given two different games with the same set of pure strategies of both players, under what conditions would they have the same set of Nash equilibria? Such a topic is generally studied under the umbrella of ``strategically equivalent games''. There could be multiple definitions of strategic equivalence of bimatrix games, depending on what we would like to preserve between the two games:
\begin{enumerate}
 \item For each player, for {\it every pure strategy} of the other player, the {\it preference ordering} of the player's pure strategies remains the same in both games (Definition \ref{def:possieri}).
 \item For each player, for {\it every mixed strategy} of the other player, the {\it sets of best response mixed strategies} of the player remain the same in both games (Definition \ref{def:liu}).
 \item For each player, for {\it every mixed strategy} of the other player, the {\it preference ordering} of the player's mixed strategies remain the same in both games (Definition \ref{def:stratEqNE}).
\end{enumerate}
It is clear that if two games are strategically equivalent according to Definition \ref{def:possieri}, then the sets of pure strategy Nash equilibria of both games are the same. On the other hand, if two games are strategically equivalent according to Definitions \ref{def:liu} and \ref{def:stratEqNE}, then the sets of all mixed strategy Nash equilibria of both games are the same. We present a brief overview of these definitions along with some consequences.

The first definition of strategically equivalent games that preserves the preference over the pure best responses of the players is proposed in \cite{hespanha2017noncooperative}.
\begin{definition}[\cite{possieri2017algebraic}, \cite{hespanha2017noncooperative}] \label{def:possieri}
Two games $(m,n,A,B)$ and $(m,n,\tilde{A},\tilde{B})$ are strongly strategically equivalent iff there exist two monotone strictly increasing functions $\phi_1,\phi_2:\Re\rightarrow\Re$ such that $\tilde a_{ij} = \phi_1(a_{ij})$ and $\tilde b_{ij} = \phi_2(b_{ij})$.
\end{definition}
As we noted above, this definition of strategic equivalence between games imply that the games have the same set of pure strategy Nash equilibria, assuming they exist.

Another definition of strategic equivalence, based on preservation of best response correspondence, is given in an unpublished paper \citep{liu1996invariance}. This notion of strategic equivalence is also discussed in \cite[p. 52-53]{myerson2013game}.
\begin{definition}[\cite{liu1996invariance}]\label{def:liu}
Two games $(m,n,A,B)$ and $(m,n,\tilde A,\tilde B)$ are strategically equivalent if and only if for any $\vct p\in\Delta_m$ and $\vct q\in\Delta_n$, we have $\Gamma_A(\vct q) = \Gamma_{\tilde A}(\vct q)$ and $\Gamma_B(\vct p) = \Gamma_{\tilde B}(\vct p)$.
\end{definition}
Another definition of strategic equivalence, based on preservation of preference ordering for all mixed strategies of the other player, was presented in \cite{moulin1978strategically}, which is given below. 
\begin{definition}[\cite{moulin1978strategically}, Def. 1, p. 205]\label{def:stratEqNE}
Two games $(m,n,A,B)$ and $(m,n,\tilde A,\tilde B)$ are strategically equivalent if and only if for any $\bar{\vct{p}},\vct p\in\Delta_m$ and $\bar{\vct{q}},\vct q\in\Delta_n$, we have
\begin{align*}
\bar{\vct{p}}^\transpose A \vct q \geq \vct p^\transpose A \vct q & \Longleftrightarrow \bar{\vct{p}}^\transpose \tilde A \vct q \geq \vct p^\transpose \tilde A \vct q\\
\vct p^\transpose B \bar{\vct{q}} \geq \vct p^\transpose A \vct q & \Longleftrightarrow \vct p^\transpose \tilde B \bar{\vct{q}} \geq \vct p^\transpose \tilde B \vct q. 
\end{align*}
\end{definition}

This notion of strategic equivalence puts more stringent requirements on the two games in comparison to the one in Definition \ref{def:liu}. Indeed, this was shown by \cite{liu1996invariance} via an example comprising two bimatrix games that have different preference orderings of the players, but have the same set of best response correspondences for both the players in the two games.

To ascertain strategic equivalence between the two games according to Definition \ref{def:stratEqNE}, we need to check that the above condition holds for all possible $\bar{\vct{p}},\vct p\in\Delta_m$ and $\bar{\vct{q}},\vct q\in\Delta_n$. This is obviously difficult using the brute-force computational approach. It turns out that two games that are strategically equivalent according to the above definition satisfy a simple property:

\begin{definition}[Positive affine transformation]\label{def:PAT}
 $(m,n,\tilde A,\tilde B)$ is a positive affine transformation (PAT) of $(m,n,A,B)$ if and only if there exists $\alpha_1,\alpha_2 \in \Rea_{>0}$, $\vct{u}\in\Rea^n$, and $\vct{v}\in\Rea^m$ such that $\tilde{A} = \alpha_1A+\vct{1}_m \vct{u}^\transpose$ and $\tilde{B} = \alpha_2B+\vct{v}\vct{1}_n^\transpose$. We define the PAT correspondence $\Upsilon:\Re^{m\times n}\times \Re^{m\times n}\rightrightarrows \Re^{m\times n}\times \Re^{m\times n}$ as follows:
\begin{align*}
 \Upsilon(A,B) = \Big\{ (\tilde A,\tilde B) \in \Re^{m\times n}\times \Re^{m\times n}: (\tilde A,\tilde B) \text{ is a PAT of } (A,B)\Big\}.
\end{align*}
\end{definition}

PAT is a commonly studied game transformation within the game theory community; see, for example, the discussion in \cite[p. 52]{myerson2013game}. It is immediate that PAT transformations preserve the preference orderings of the players for any given strategy of the other player; consequently, if two games are a PAT of each other, then they are strategically equivalent according to Definition \ref{def:stratEqNE}. The surprising result in \cite{moulin1978strategically} is that in fact, the converse also holds.

\begin{lemma}\label{lem:stratEqVec}
	Two games $(m,n,A,B)$ and $(m,n,\tilde A,\tilde B)$ are strategically equivalent according to Definition \ref{def:stratEqNE} if and only if they are PAT of each-other.
\end{lemma}
\begin{proof}
This result is established in \cite[Theorem 1, p. 206]{moulin1978strategically}.
\qed
\end{proof}
It is also straightforward to verify the following result.

\begin{lemma}
If two games are strategically equivalent according to Definition \ref{def:stratEqNE}, then they are strategically equivalent according to Definition \ref{def:liu}.
\end{lemma}

Given that PAT and the notion of strategic equivalence due to Definition \ref{def:stratEqNE} are equivalent, we will use Definition \ref{def:PAT} throughout this text. It is clear from the above discussion that the PAT operation can transform the rank of the game. In this paper, we ask the following question: Given a game $(m,n,\tilde{A},\tilde{B})$, is there an algorithm that determines the parameters $(\alpha_1,\alpha_2,\vct{u},\vct{v})$, such that $(m,n,\tilde{A},\tilde{B})$ is strategically equivalent to a lower rank game $(m,n,A,B)$ via a PAT? In the process, we solve the following problem:
\begin{align}\label{eqn:probFor}
 (\hat A,\hat B) \in \argmin_{(A,B)\in \Upsilon(\tilde A,\tilde B)} \rank{A+B}.
\end{align}
We proceed to solving this problem in the next section.

\section{Main Result}\label{sec:MainResult}
From Definition \ref{def:PAT}, it is straightforward to derive the following algebraic relationship for a PAT: If there exists $\alpha_1,\alpha_2 \in \Re_{>0}$, $\vct{u}\in\Re^n$, and $\vct{v}\in\Re^m$ such that:
\begin{align}
    \tilde{A} &= \alpha_1A+\vct{1}_m \vct{u}^\transpose, \label{eq:AtildeRankk}\\
    \tilde{B} &= -\alpha_2A+\vct{v}\vct{1}_n^\transpose +\alpha_2\sum_{i=1}^{k}\vct{r}_i\vct{c}_i^\transpose \label{eq:BtildeRankk}
\end{align}
then $(m,n,\tilde{A},\tilde{B})$ is strategically equivalent  to the rank-$k$ game $(m,n,A,-A+\sum_{i=1}^{k}\vct{r}_i\vct{c}_i^\transpose)$ via a positive affine transformation (PAT). To see this, we combine \eqref{eq:AtildeRankk} and \eqref{eq:BtildeRankk} to get

\begin{align}
	\tilde{A} &= -\frac{\alpha_1}{\alpha_2} \tilde{B}+\alpha_1\sum_{i=1}^{k}\vct{r}_i\vct{c}_i^\transpose+ \vct{1}_m \vct{u}^\transpose+\frac{\alpha_1}{\alpha_2}\vct{v}\vct{1}_n^\transpose. \label{eq:AtildeRank1Combined}
\end{align}

Defining $\gamma\coloneqq\frac{\alpha_1}{\alpha_2}$, $\hat{\vct{u}}\coloneqq\vct{u}$, $\hat{\vct{v}}\coloneqq\gamma\vct{v}$, and letting $\alpha_1\vct{r}_i\vct{c}_i^\transpose=\hat{\vct{r}}_i\hat{\vct{c}}_i^\transpose$, we rewrite \eqref{eq:AtildeRank1Combined} as:
\begin{equation}\label{eq:AtildeGammaBtilde}
    \tilde{A}+\gamma\tilde{B}=\vct{1}_m\hat{\vct{u}}^\transpose+\hat{\vct{v}}\vct{1}_n^\transpose+\sum_{i=1}^{k}\hat{\vct{r}}_i\hat{\vct{c}}_i^\transpose.
\end{equation}

From \eqref{eq:AtildeGammaBtilde}, we see that any algorithm that determines the parameters $(\gamma,\vct{\hat{u}},\vct{\hat{v}})$ solves the problem of finding $(A,B)\in \Upsilon(\tilde A,\tilde B)$ such that $\rank{A+B}\leq\texttt{rank} (\tilde{A}+\tilde{B})$. The uniqueness of our algorithm that we present in this section is that it solves problem \eqref{eqn:probFor}. In other words, we find a game, $(m,n,\hat{A},\hat{B})$, that is both a minimum rank game and a PAT of the original game, $(m,n,\tilde{A},\tilde{B})$.

Our approach to solving problem \eqref{eqn:probFor} is to decompose it into two parts based on solving $(A,B)\in \Upsilon(\tilde A,\tilde B)$ and a third part based on constructing the minimum rank game, $(m,n,\hat{A},\hat{B})$. These three parts are:

\begin{enumerate}
    \item Given $(m,n,\tilde{A},\tilde{B})$, we apply the Wedderburn rank reduction formula to $\tilde A$ and $\tilde B$ to arrive at matrices $\bar A$ and $\bar B$ such that $\vct 1_m\not\in\colspan{\bar A},\colspan{\bar B}$ $\vct 1_n\not\in\colspan{\bar A^\transpose},\colspan{\bar B^\transpose}$.\label{prob:findD}
    \item We use matrix pencil theory to compute $\gamma^*\in\Re_{>0}$ such that $\rank{\bar{A}+\gamma^*\bar{B}}$ is as small as possible. \label{prob:findGamma}
    \item Using $\gamma^*$, we construct the strategically equivalent game $(m,n,\hat{A},\hat{B})$ \label{prob:findhatAhatB}
\end{enumerate}

We now introduce the Wedderburn rank reduction formula and show how that can reduce the rank of a game by at most two. Following that, we discuss matrix pencils, an existing canonical form for calculating the eigenvalues of rectangular pencils, and show how such a canonical form can be applied to obtain a bimatrix game of lower rank than the original game. Finally, we conclude the section with the statement and proof of our main result.

\subsection{Rank Reduction via the Wedderburn Rank Reduction Formula}\label{subsec:furtherReduction}

The Wedderburn Rank Reduction formula is a classical technique in linear algebra that allows one to reduce the rank of a matrix by subtracting a specifically formulated rank-1 matrix. By repeated applications of the formula, one can obtain a matrix decomposition as the sum of multiple rank-1 matrices. In contrast to other well-known matrix factorization algorithms, such as singular value decomposition, the Wedderburn rank reduction formula allows almost limitless flexibility in choosing the basis of the rank-1 matrices that are subtracted at each iteration. For further reading on the Wedderburn rank reduction formula, we refer the reader to \citet[p.~69]{wedderburn1934lectures}, or to the excellent treatment of the topic by \cite{chu1995rank}.

We now proceed to state Wedderburn's original theorem. Following that, we show how one can exploit the flexibility of the decomposition to extract specifically formulated rank-1 matrices that allow us, when certain conditions hold true, to reduce the rank of a bimatrix game.

\begin{theorem}[{\citet[p.~69]{wedderburn1934lectures},\citet{chu1995rank}}] \label{thm:wedderburn}
	Let $C\in\Re^{m\times n}$ be a non-zero matrix. Then,  the following holds:
	\begin{enumerate}
	\item There exists vectors $\vct{x}_1\in \Re^n$ and $\vct{y}_1\in \Re^m$ such that $w_1=\vct{y}_1^\transpose C \vct{x}_1\neq0$. Then, the matrix
	\begin{equation} \label{eq:Wedderburn}
	C_2 \coloneqq C-w_1^{-1}C\vct{x}_1\vct{y}_1^\transpose C
	\end{equation}
	has rank exactly one less than the rank of $C$.
	 \item A vector $\vct z\in\colspan{C_2}$ if and only if $\vct z\in \colspan{C}$ and $\vct z\perp \vct y_1$. Similarly, $\vct z\in\colspan{C_2^\transpose}$ if and only if $\vct z\in \colspan{C^\transpose}$ and $\vct z\perp \vct x_1$.
	 \item We have $C\vct x_1\not\in \colspan{C_2}$ and $C^\transpose\vct y_1\not\in\colspan{C_2^\transpose}$.
    \end{enumerate}
\end{theorem}
\begin{proof}
The proof of the first assertion is due to \citet[p.~69]{wedderburn1934lectures}. It essentially shows that the null space of $C_2$ is equal to the span of the null space of $C$ and $\vct x_1$. The second assertion is established in Proposition \ref{prop:colSpanC2} in the appendix. The third assertion is a consequence of the second assertion and the fact that $\vct y_1^\transpose C\vct x_1\neq 0$. 
\qed
\end{proof}

We now use the Wedderburn rank reduction formula whenever there exists $\vct x_1\in\Re^n$ such that $C\vct{x}_1 = \vct 1_m$ (equivalently $\vct 1_m\in\colspan{C}$) and/or there exists $\vct y_1\in\Re^m$ such that $C^\transpose \vct y_1 = \vct 1_n$). This operation can reduce the rank by at most 2. We outline this operation next.

For a matrix $D\in\Re^{m\times n}$, recall that $D_{(i)}$ is the $i^{th}$ row (represented as a column vector) and $D^{(j)}$ is the $j^{th}$ column. Let us define a map $\Psi:\Re^{m\times n}\rightarrow \Re^n\times\Re^m\times\{0,1,2\}$ as follows: For a matrix $D\in\Re^{m\times n}$, $\Psi(D)$ is defined as the tuple $({\vct u},{\vct v},l)$, which is defined as
\begin{enumerate}
 \item If $\vct 1_m \not\in\colspan{D}$ and $\vct 1_n \not\in\colspan{D^\transpose}$, then ${\vct u} = \vct 0_n$, ${\vct v} = \vct 0_m$, and $l = 0$.
 \item If $\vct 1_m \in\colspan{D}$, then ${\vct u} = D_{(i)}$ for any $i\in\{1,\ldots,m\}$, ${\vct v} = \vct 0_m$, and $l = 1$.
 \item If $\vct 1_n \in\colspan{D^\transpose}$, then  ${\vct u} = \vct 0_n$, ${\vct v} = D^{(j)}$ for any $j\in\{1,\ldots,n\}$, and $l = 1$.
 \item For the case $\vct 1_m \in\colspan{D}$, $\vct 1_n \in\colspan{D^\transpose}$, we have two possibilities:
 \begin{enumerate}
   \item If the nullspace of the matrix $\begin{bmatrix} D & \vct{1}_m\\ \vct{1}^\transpose_n & 0  \end{bmatrix}$ comprises of vectors of the form $\mat{\vct x\\0}$, then define ${\vct u} = D_{(i)}$ for any $i\in\{1,\ldots,m\}$, ${\vct v} = \vct 0_m$, and $l = 1$.
   \item On the other hand, if there exists a vector of the form $\mat{\vct x\\-1}$ in the nullspace of the matrix $\begin{bmatrix} D & \vct{1}_m\\ \vct{1}^\transpose_n & 0  \end{bmatrix}$, then  ${\vct u} = D_{(i)}$ and ${\vct v} = (D^{(j)} - d_{ij}\vct 1_n)$ for any $i\in\{1,\ldots,m\}$ and $j\in\{1,\ldots,n\}$, and $l = 2$.
 \end{enumerate}
\end{enumerate}
It is clear from the preceding discussions on the Wedderburn rank reduction formula that the following holds.

\begin{lemma}\label{lem:Drankred}
 For any matrix $D$, let $({\vct u},{\vct v},l) = \Psi(D)$. Then, the matrix $\bar D = D - \vct 1_m {\vct u}^\transpose - {\vct v} \vct 1_n^\transpose$ has the rank $(\rank{D}-l)$. Further, $\vct 1_m\not\in\colspan{\bar D}$ and $\vct 1_n\not\in\colspan{\bar D^\transpose}$.
\end{lemma}
\begin{proof}
 A proof is presented in Appendix \ref{app:Drankred}.
\end{proof}

We now define another matrix transformation and two additional lemmas that are necessary when we prove that $(m,n,\hat A,\hat B)$ is a strategically equivalent game of lowest rank in Subsection \ref{sub:MainResult}. Let $M_k$ be a square matrix defined as 
\begin{align}\label{eqn:Mk}
M_k := \left(I_k - \frac{1}{k}\vct 1_{k\times k}\right), 
\end{align}
where $I_k$ is an identity matrix of dimension $k\times k$. We note here that $\vct 1_k$ is in the nullspace of $M_k$.  For a matrix $D\in\Re^{m\times n}$, define $\vct u, \vct v$ and $\tilde D$ as
 \begin{align*}
  \vct u & := \frac{1}{m}D^\transpose \vct 1_m, \; \vct v := \frac{1}{n} M_m D \vct 1_n,\\
  \tilde D & := D - \vct 1_m \vct u^\transpose - \vct v \vct 1_n^\transpose = M_m D M_n.
 \end{align*}
Define $\Xi:\Re^{m\times n}\rightarrow \Re^{m\times n}\times \Re^n\times\Re^m$ as
\begin{align}
 \Xi(D) := (\tilde D, \vct u,\vct v).
\end{align}
We show in the lemma below that the rows of the matrix $\tilde D$ are orthogonal to $\vct 1_n$ and the columns of the matrix $\tilde D$ are orthogonal to $\vct 1_m$.
\begin{lemma}
 For a matrix $D\in\Re^{m\times n}$, let $(\tilde D, \vct u,\vct v):=\Xi(D)$. Then, the following holds: (a) $\vct 1_m^\transpose \tilde D = 0$, $\tilde D\vct 1_n = 0$ and (b) $\vct 1_m\not\in\colspan{\tilde D}$, $\vct 1_n \not \in\colspan{\tilde D^\transpose}$.
\end{lemma}
\begin{proof}
 The proof of part (a) and (b) follows from simple algebraic manipulations and is therefore omitted. Further, since $\vct 1_m$ is orthogonal to every column of $\tilde D$, $\vct 1_m$ cannot be in the column span of $\tilde D$. Analogously, $\vct 1_n \not \in\colspan{\tilde D^\transpose}$. This completes the proof. \qed
\end{proof}

The matrix transformation induced by the map $\Xi$ does not reduce the rank of the matrix beyond what is possible through the map $\Psi$. This result is formally established below.

\begin{lemma}\label{lem:Dranksame}
Let $D\in\Re^{m\times n}$ be a matrix.  Let $({\vct u},{\vct v},l) = \Psi(D)$ and define $\bar D = D - \vct 1_m {\vct u}^\transpose - {\vct v} \vct 1_n^\transpose$. Define $(\tilde D, \tilde{\vct u},\tilde{\vct v}):=\Xi(\bar D)$. Then, $\tilde D = M_m D M_n$ and $\rank{\tilde D} = \rank{\bar D} = \rank{D} - l$.
\end{lemma}
\begin{proof}
We have $\tilde D = M_m\left(D - \vct 1_m {\vct u}^\transpose - {\vct v} \vct 1_n^\transpose\right)M_n =  M_m D M_n$. Let $k$ denote the rank of $\bar D$. Recall from Theorem \ref{thm:wedderburn} part 3 that $\vct 1_m\not\in\colspan{\bar D}$ and $\vct 1_n\not\in\colspan{\bar D^\transpose}$. Since $\vct 1_m\not\in\colspan{\bar D}$ and the nullspace of $M_m$ is of the form $\xi\vct 1_m$, $M_m \bar D$ is a rank $k$ matrix. Further, since $\vct 1_n\not\in\colspan{\bar D^\transpose}$ and $\colspan{\bar D^\transpose M_m}\subset \colspan{\bar D^\transpose}$, we conclude that $\vct 1_n\not\in\colspan{\bar D^\transpose M_m}$. Since $\vct 1_n\not\in\colspan{\bar D^\transpose M_m}$ and the nullspace of $M_n$ is of the form $\xi \vct 1_n$, we have $M_n\bar D^\transpose M_m$ is also a rank $k$ matrix.
\qed
\end{proof}

We next proceed to a discussion on the matrix pencil theory and identify how to apply it to reduce the rank of the game.

\subsection{Matrix Pencils}\label{subsec:matrixPencils}
Let $A,B\in\Re^{m\times n}$ be matrices of known values, and let $\lambda$ represent an unknown parameter. Then the set of all matrices of the form $A+\lambda B$, with $\lambda\in\Ce$, define a \textit{linear matrix pencil} (or just a \textit{pencil})\citep[p.~24]{ikramov1993matrix,gantmakher1959theory}. The literature defines both the set of matrices $A+\lambda B$ and $A-\lambda B$ as pencils. Although $A-\lambda B$ seems to be more common (possibly due to the connection to the standard presentation of the eigenvalue problem $A-\lambda I$), we choose to use $A+\lambda B$ as it more closely aligns with the problem presented in \eqref{eqn:probFor}.

While not as well studied as the standard eigenproblem, $A-\lambda I$, the theory of pencils still enjoys a rich history. For the square, nonsingular, $m\times m$ case, Weierstrass investigated pencils and developed a canonical form as early as 1867. The rectangular case was later solved, with a canonical form presented, by Kronecker in 1890. His canonical form, aptly named the Kronecker Canonical Form (KCF), was popularized by Gantmacher in chapter XII of his two volume treatise on the theory of matrices \cite[Ch.~12]{gantmakher1959theory}. For a (not so short) survey on pencils, we refer the reader to \citep{ikramov1993matrix}. For a discussion of the various canonical forms and computational methods, for the square case see \cite[Ch.~7.7]{golub2012matrix} and, for the more general singular/rectangular case, see \citep{demmel1993generalized}.

In a series of papers, \cite{weil1968game,thompson1972roots}, the authors study the relationship between the eigenvalues and eigenvectors of matrix pencils and the solution of a zero-sum game, where the game is formulated as $(m,n,A-\lambda B,-A+\lambda B)$. The square, traditional eigenvalue problem with $(m,m,A-\lambda I,-A+\lambda I)$ is studied in \cite{weil1968game}. The authors present the rectangular matrix pencil version in \cite{thompson1972roots}. Their results are indeed theoretically interesting; however, as \cite{weil1968game} states, the relationship between eigensystems and game theory is ``tenuous". This seems to make, at least in accordance with the current theoretical results, the study of eigensystems ill-suited as a solution concept for bimatrix games. In contrast, as we will show in this subsection, the matrix pencil problem is well-suited to the study of strategically equivalent games. 

In the remainder of this section, we review the canonical form of a matrix pencil presented by  \cite{thompson1970reducing,thompson1972roots}. Although not a common terminology in the literature, we'll refer to this canonical form as the Thompson-Weil Canonical Form (TWCF). Our motivation for studying the TWCF is two-fold.  First off, the TWCF focuses on only computing those eigenvalues, if they exist, that strictly reduce the rank of the pencil $A+\lambda B$. Other extraneous values, such as those computed in the KCF, are ignored. Secondly, we wish to bring renewed emphasis on existing results that connect the study of matrix pencils to game theory.

For completeness, we now restate some results from \citep{thompson1970reducing,thompson1972roots,dell1971algorithm}. Following that, we show how to apply those results to calculate an equivalent game of lower rank. 
\smallskip

For the following discussions, let $\rank{A}=p$ and $\rank{B}=r$.

\begin{definition}[{\citet[Definition~2.1]{thompson1972roots}}]\label{def:pencilSolution}
    By a solution to the pencil $A+\lambda B$, we shall mean a triple $(\lambda,\vct{x},\vct{y})\in\Ce\times\Re^n\times\Re^m$, satisfying $\vct{x}\neq\vct{0}$ and $\vct{y}\neq\vct{0}$, that solve the set of equations
    \begin{align*}
        (A+\lambda B)\vct{x}=\vct{0},\qquad  \vct{y}^\transpose(A+\lambda B)=\vct{0},
    \end{align*}
    and have the property that $\rank{A+\lambda B}<\rank{A+\mu B}$ for any $\mu$ that is not an element of the solution triple.
\end{definition}
Throughout their series of works on matrix pencils,  \cite{thompson1970reducing,thompson1972roots,dell1971algorithm} refer to the solution triple in Definition \ref{def:pencilSolution} by various names such as \textit{pencil value}, \textit{pencil roots}, \textit{rank-reducing numbers}, \textit{left pencil-vector}, and \textit{right pencil-vector}. For simplicity, we choose to use the terms \textit{eigenvalue} and \textit{left\textbackslash right eigenvector}. We will also use the set $\Lambda(A,B)$ to represent all $\lambda$ that are in the solution triple as defined in Definition \ref{def:pencilSolution}. This set $\Lambda(A,B)$ exactly corresponds to the eigenvalues of the Jordan block in KCF form.    

\begin{lemma}[{\citet[Lemma~1]{thompson1970reducing}}]\label{lem:TW1}
For all $A,B\in\Re^{m\times n}$, there exists nonsingular $S_1$ and $T_1$ such that $S_1(A+\lambda B)T_1$ can be partitioned into:
\begin{equation}\label{eq:lem:TW1}
\kbordermatrix{  &r       &t        &q\\
                r&E_{11}  &  E_{12} & 0 & 0\\
                s& E_{21} & 0       & 0 & 0\\
                q& 0      & 0       & I_q & 0\\
                & 0       & 0       & 0   & 0} \\
                +\lambda
 \kbordermatrix{&       \\
                &I_r  &  0 & 0 & 0\\
                & 0   & 0  & 0 & 0\\
                & 0   & 0  & 0 & 0\\
                & 0   & 0  & 0 & 0} \\
\end{equation}
where $q\leq\min\{m-r,n-r,p\}$, $E_{12}$ is in column echelon form, $\rank{E_{12}}=t$, $E_{21}$ is in row echelon form, and $\rank{E_{21}}=s$. Any of $s,t,q$ may be zero, $r+s+q\leq m$, and $r+t+q\leq n$. 
\end{lemma}
\begin{proof}
    See \citep[Lemma~1]{thompson1970reducing} or \citep[Theorem~2.1]{thompson1972roots}.
\end{proof}

\begin{lemma}[{\citet[Lemma~2]{thompson1970reducing}}]\label{lem:TW2}
If $s+t>0$ in Lemma \ref{lem:TW1}, there exists nonsingular $S_2$ and $T_2$ such that $S_2(A+\lambda B)T_2$ can be partitioned into: 
\begin{equation}\label{eq:lem:TW2}
\kbordermatrix{  &r-s     &s    &t        &q\\
              r-t&C_{11}  &  0 & 0 & 0 &0\\
                t& 0      & 0  & I_t &0 &0\\       
                s& 0      & I_s& 0   &0 & 0\\
                q& 0      & 0  & 0   & I_q & 0\\
                & 0       & 0  &0    & 0   & 0} \\
                +\lambda
 \kbordermatrix{&       \\
                &D_{11}  & D_{12} & 0 & 0 &0\\
                &D_{21}  & D_{22}  & 0 & 0&0\\
                & 0   & 0  & 0 & 0 &0\\
                & 0   & 0  & 0 & 0 &0\\
                & 0   & 0  & 0 & 0 &0} \\
\end{equation}
\end{lemma}
\begin{proof}
 See \cite[Lemma~2]{thompson1970reducing}.
\end{proof}
Furthermore, by repeated applications of Lemmas \ref{lem:TW1} and \ref{lem:TW2}, the authors define an iterative algorithm that solves for the set $\Lambda(A,B)$, including identifying if $\Lambda(A,B)=\emptyset$. We briefly outline the algorithm in Algorithm \ref{alg:TWCF}. For the full proof and implementation details, we refer the reader to \citep{thompson1970reducing,dell1971algorithm}.
\begin{remark}
As the authors note in \cite{dell1971algorithm}, their algorithm may be numerically unsound for ill-conditioned problems. Indeed, while there does not appear to be any results in the literature comparing the numerical stability of the TWCF algorithm and Gantmacher's method for computing the KCF, it seems likely that both methods may share similar numerical difficulties. Therefore, other algorithms may be better suited for ill-conditioned problems. For the square, dense matrix pencil, (whether ill-conditioned or not) the famous QZ algorithm of \cite{moler1973algorithm} is likely a better option. One can also see \cite[Ch.~7.7]{golub2012matrix} for a detailed description of the QZ algorithm. For the rectangular case, numerical accuracy for ill-conditioned problems can likely be improved via the GUTPRI algorithm developed by \cite{demmel1993generalized}. However, even in light of this discussion, we choose to explore the TWCF as it provides insight into the mathematical structure of the pencil that applies to our problem at hand.  
\end{remark}

\begin{algorithm}[tbh]
	\caption{Algorithm for computing the eigenvalues of a pencil using the Thompson-Weil Canonical Form}
	\label{alg:TWCF}
	\begin{algorithmic}[1] 
		\Function{TWCF}{$A,B$}
        \State $i\gets1$
        \State $A_1\gets A,B_1\gets B$
        \State flag$\gets$True
		\While{flag}
		\State Calculate $E_{11,i},I_{r_i},r_i,s_i,t_i,q_i$ via \eqref{eq:lem:TW1} 
		\If{$r_i=s_i\lor r_i=t_i$}
		    \State $\Lambda(A,B)\gets\emptyset$
		    \State flag$\gets$False
		 \ElsIf{$s_i+t_i=0$}
		     \State $\Lambda(A,B)\gets\Lambda(E_{11,i},I_{r_i})$
		     \State  flag$\gets$False
		  \Else
		      \State Calculate $C_{11,i},D_{11,i}$ via \eqref{eq:lem:TW2} 
		      \State $A_i\gets C_{11,i},B_i\gets D_{11,i}$ 
		\EndIf
		\EndWhile
        \State \textbf{return} ($\Lambda(A,B)$)
        \EndFunction
	\end{algorithmic}
\end{algorithm}
    
We now state one final definition before proceeding to state the main theorem from \cite{thompson1970reducing}, which we use to prove our next result.

\begin{definition}\label{def:mLambda}
    For any $\lambda\in\Ce$, define the geometric multiplicity of $\lambda$ as:
    \begin{enumerate}
        \item $\m{\lambda}=0$ if $\lambda\notin\Lambda(A,B)$,
        \item $\m{\lambda}=$ the number of Jordan blocks containing $\lambda$ in the Jordan normal form of $E_{11,i}$ where $i$ is the smallest integer such that $s_i+t_i=0$.
    \end{enumerate}
\end{definition}

\begin{theorem}[{\citet[Theorem~1]{thompson1970reducing}}]\label{thm:TWCFthm1}
    For any complex number $\lambda$,
    \begin{equation*}
        \rank{A+\lambda B}=r+q-\m{\lambda}
    \end{equation*}
    where $r$ and $q$ are defined in Lemma \ref{lem:TW1} and $\m{\lambda}$ is as defined in Definition \ref{def:mLambda}.
\end{theorem}
\begin{proof}
    See \cite[Theorem~1]{thompson1970reducing}.
\end{proof}

As we are only concerned with the real, strictly positive eigenvalues, let us further define the restricted set of eigenvalues as:
\begin{equation*}
    \Lambda_{>0}(A,B)=\{\lambda\in \Re_{>0}\vert \lambda\in\Lambda(A,B)\} = \Lambda(A,B)\cap (0,\infty).
\end{equation*}
Let $\gamma^*(A,B)$ be defined as 
\begin{align}\label{eqn:gammastar}
 \gamma^*(A,B) = 
 \begin{cases}
        \argmax_{\lambda\in\Lambda_{>0}(A,B)} \m{\lambda} & \text{if } \Lambda_{>0}(A,B)\neq \emptyset\\
        1 & \text{otherwise}.
 \end{cases}
\end{align}

\begin{lemma}\label{lem:ABpencil}
 For matrix $(A,B)$, let $\gamma^*:=\gamma^*(A,B)$. Then, $\rank{A+\gamma^* B} \leq \rank{A+\lambda B}$ for all $\lambda >0$. Further, $\rank{A+\gamma^*B} =  \rank{A+ \lambda B} +\m{\lambda}-\m{\gamma^*}$ for all $\lambda>0$.
\end{lemma}
\begin{proof}
    Clearly, if $\lambda_{>0}(A,B)=\emptyset$, then for all $\lambda\in\Re_{>0}$ we have $\rank{A+\lambda B}=r+q$ by Theorem \ref{thm:TWCFthm1}. Thus, the statement holds trivially.
    
    Let us suppose that $\lambda_{>0}(A,B)\neq\emptyset$. Since $\m{\gamma^*}\geq \m{\lambda}$ for all $\lambda>0$, we get
    \begin{equation*}
        \rank{A+\gamma^*B}=r+q-\m{\gamma^*}\leq r+q-\m{\lambda}=\rank{A+\lambda B}.
    \end{equation*}
    This yields the result.
    \qed
\end{proof}

\subsection{The Algorithm for Rank Reduction and the Main Result}\label{sub:MainResult}
Finally, we now state our main result. Consider the game $(m,n,\tilde{A},\tilde{B})$ with $\rank{\tilde{A}+\tilde{B}}=\tilde{k}\geq1$. Define the rank reduction map $\Pi:\Re^{m\times n}\times \Re^{m\times n}\to \Re^{m\times n}\times \Re^{m\times n}$ by employing the following five step process. 
\begin{enumerate}
 \item Let $(\tilde{\vct u}_{\tilde A},\tilde{\vct v}_{\tilde A},l_{\tilde A}):=\Psi(\tilde A)$ and $(\tilde{\vct u}_{\tilde B},\tilde{\vct v}_{\tilde B},l_{\tilde B}):=\Psi(\tilde B)$. 
 \item Compute $(A^\dagger,\vct u^\dagger_A,\vct v^\dagger_A) := \Xi(\tilde A)$ and $(B^\dagger,\vct u^\dagger_B,\vct v^\dagger_B) := \Xi(\tilde B)$.
\item Compute $\gamma^*:=\gamma^*(A^\dagger,B^\dagger)$. Define $\hat{\vct u}$, $\hat{\vct v}$, and $\hat l$ as
\begin{align*}
 \hat{\vct u}:=\tilde{\vct u}_{\tilde A}+\gamma^*\tilde{\vct u}_{\tilde B} 
 \quad \hat{\vct v}:=\tilde{\vct v}_{\tilde A}+\gamma^*\tilde{\vct v}_{\tilde B}, 
 \quad \hat l:=\begin{cases}
     0 & \text{if } \hat{\vct u} = \vct 0_n, \hat{\vct v} =\vct 0_m\\
     2 & \text{if } \hat{\vct u} \neq \vct 0_n, \hat{\vct v} \neq \vct 0_m\\
     1 & \text{otherwise}
    \end{cases}.
\end{align*}
\item Determine $\check A := \tilde A - \vct{1}_m\hat{\vct{u}}^\transpose$, $\check B := \gamma^*\tilde B -\hat{\vct{v}}\vct{1}_n^\transpose$, and $\check C := \check A+ \check B$. Let $(\check{\vct u},\check{\vct v},\check l) := \Psi(\check C)$.
\item Define $\Pi\Big(\tilde A,\tilde B\Big) := \Big(\hat A,\hat B\Big)$, where $\hat A = \check A - \vct{1}_m\check{\vct{u}}^\transpose$ and $\hat B = \check B -\check{\vct{v}}\vct{1}_n^\transpose$. 
\end{enumerate}

\begin{theorem}\label{thm:theMainResult}
The game $(m,n,\hat{A},\hat{B})$ is strategically equivalent to the game $(m,n,\tilde{A},\tilde{B})$ with 
\begin{align}\label{eqn:truerankab}
\rank{\hat{A}+\hat{B}}=\rank{\tilde{A}+\tilde{B}} +\m{1}-\m{\gamma^*}-\hat l - \check l.
\end{align}
Further, $(\hat A,\hat B)$ solves \eqref{eqn:probFor}.
\end{theorem}
\begin{proof}
 The strategic equivalence of games $(m,n,\hat{A},\hat{B})$ and $(m,n,\tilde{A},\tilde{B})$ follows from Lemma \ref{lem:stratEqVec}. Equation \ref{eqn:truerankab} follows from the Wedderburn rank reduction formula in Theorem \ref{thm:wedderburn}, Lemma \ref{lem:Drankred}, and Lemma \ref{lem:ABpencil}. We only need to show that $(\hat A,\hat B)$ solves \eqref{eqn:probFor}.
 
 In Step 3, we have $\rank{A^\dagger+\gamma^* B^\dagger} \leq \rank{A^\dagger+\lambda B^\dagger}$ for all $\lambda>0$. Lemma \ref{lem:Dranksame} implies that 
 \begin{align*}
  \rank{\hat A+\hat B} = \rank{M_m (\hat A+\hat B) M_n}= \rank{A^\dagger+\gamma^* B^\dagger},
 \end{align*}
 where $M_k$ is defined in \eqref{eqn:Mk}. Thus, $(m,n,\hat A,\hat B)$ is a strategically equivalent game of lowest rank.\qed
\end{proof}

We now present some examples where we show that each step of the above algorithm is important in the rank reduction process of the game. If we miss any of the steps, then the algorithm may not be able to identify the game of lowest rank. For the following examples, pick $\vct r,\vct v\in\Re^m$ such that $\vct r,\vct v\neq \vct 1_m$ and $\vct r\neq \vct v$. Similarly, pick $\vct c,\vct u\in\Re^n$ such that $\vct c,\vct u\neq \vct 1_n$ and $\vct c\neq \vct u$. Let $D := \vct 1_m \vct u^\transpose + \vct v\vct 1_n^\transpose$.

\begin{example}
Consider the game with payoff matrices $\tilde A = 2\vct r \vct c^\transpose +D$, $\tilde B = -\vct r \vct c^\transpose -D$. Here, a direct application of matrix pencil theory will yield $\{1,1,2\}\subset\Lambda_{>0}(\tilde A,\tilde B)$. Thus, if we apply the rank reduction formula without Step 1, then $\check A = 2\vct r \vct c^\transpose,\check B = -\vct r \vct c^\transpose$, which leads to a rank-1 game. However, if we apply Step 1, then $\bar A = 2\vct r \vct c^\transpose, \bar B = -\vct r \vct c^\transpose$, which will eventually lead to a strategically equivalent rank-0 game.
\end{example}

\begin{example}
Consider the game with payoff matrices $\tilde A = 2\vct r (\vct c+\vct 1_n)^\transpose$, $\tilde B = -\vct r \vct c^\transpose$, which is a rank 1 game. Then, after Step 1, we get $\bar A = \tilde A$ and $\bar B = \tilde B$. If we skip Step 2 and directly jump to Step 3, then one can have $2\not\in\Lambda_{>0}(\bar A,\bar B)$. Step 2 of the algorithm is needed so that we have $2\in\Lambda_{>0}(A^\dagger,B^\dagger)$. This leads to a strategically equivalent rank 0 game. 
\end{example}

\begin{example}
Assume that $\vct r\perp \vct 1_m$ and $\vct c\perp \vct 1_n$. Consider the game with payoff matrices $\tilde A = 2\vct r (\vct c+\vct 1_n)^\transpose + \vct 1_m\vct u^\transpose$, $\tilde B = -(\vct r+\vct 1_m) \vct c^\transpose + \vct v\vct 1_n^\transpose $, which is a rank 3 game. The output of each step of the algorithm is as follows:
\begin{enumerate}
 \item After Step 1, we get $\bar A = 2\vct r (\vct c+\vct 1_n)^\transpose$ and $\bar B = -(\vct r+\vct 1_m) \vct c^\transpose $.
 \item After Step 2, we get $A^\dagger = 2\vct r \vct c^\transpose$ and $B^\dagger = -\vct r\vct c^\transpose$.
 \item After Step 3, we get $\gamma^\star = 2$.
 \item After Step 4, we get $\check A = 2\vct r \vct c^\transpose$ and $\check B = -2\vct r\vct c^\transpose$.
\end{enumerate}
Thus, we arrive at a strategically equivalent zero-sum game.
\end{example}

We further note here that the reason for taking the five step process is to determine the rank of the game in \eqref{eqn:truerankab}, which requires the application of the Wedderburn rank reduction formula to determine $\hat l$ and $\check l$. If we are not interested in determining the rank of the game, then we can follow the following three step process:
\begin{enumerate}
\item Compute $(A^\dagger,\vct u^\dagger_A,\vct v^\dagger_A) = \Xi(\tilde A)$ and $(B^\dagger,\vct u^\dagger_B,\vct v^\dagger_B) = \Xi(\tilde B)$.
\item Compute $\gamma^*:=\gamma^*(A^\dagger,B^\dagger)$. Define $\hat{\vct u}$ and $\hat{\vct v}$ as
\begin{align*}
 \hat{\vct u}:=\vct u_A^\dagger+\gamma^*\vct u^\dagger_B 
 \quad \hat{\vct v}:=\vct v_A^\dagger+\gamma^*\vct v^\dagger_B.
\end{align*}
\item Output $\Pi^\dagger\Big(\tilde A,\tilde B\Big) := \Big(\hat A,\hat B\Big)$, where $\hat A = \tilde A - \vct{1}_m\hat{\vct{u}}^\transpose$ and $\hat B = \gamma^*\tilde B -\hat{\vct{v}}\vct{1}_n^\transpose$.
\end{enumerate}
This map $\Pi^\dagger$ outputs a game of the same rank as that of the rank reduction map $\Pi$, but in this case we cannot compute the rank of the game $(\hat A,\hat B)$ using the expression in \eqref{eqn:truerankab} since $\hat l$ and $\check l$ are unknown.

\section{A Fast Algorithm for Solving Rank-1 Matrix Pencils}\label{sec:rank1Pencils} 

With the recent publication of a polynomial time algorithm for solving rank-1 games in \cite{adsul2011rank}, it seems interesting to ask the question: given game $(m,n,\tilde{A},\tilde{B})$, does there exist $(\hat{A},\hat{B})\in \Upsilon(\tilde A,\tilde B)$ such that $\rank{\hat{A}+\hat{B}}=1$? Our five step process in Subsection \ref{sub:MainResult} could answer this question. However, as we will show in Subsection \ref{subsec:algorithm}, solving the matrix pencil problem, for $m\leq n$, has worst case complexity of $\bigO(n^4)$ operations.   

Therefore, in this section, we develop a series of results that allow us to determine whether or not there exists a $\lambda^*\in\Ce$ such that $\rank{A+\lambda^*B}=1$ by solving for the roots of a single polynomial equation (at worse a quadratic) and then conducting at most two matrix comparisons. This approach dramatically speeds up solving the matrix pencil problem for the rank-1 case. Let us begin by stating some facts about rank-1 matrices that will allow us to easily ascertain when a given matrix is rank-$1$ and to solve for values of $\lambda^*$, when they exist, such that the matrix pencil, $A+\lambda^* B$, is a rank-$1$ pencil.

\begin{fact}\label{fact:rank1facts}
The matrix $M$ in $\Re^{m\times n}$ is rank-$1$ if and only if $M\neq\vct{0}_{m\times n}$ and the following hold true:
\begin{enumerate}
    \item Every row (column) of $M$ is a scalar multiple of every other row (column) of $M$.
    \item Choose any element $m_{i,j}$ of $M$ such that $m_{ij}\neq0$ and form $\vct{r}_j=M^{(j)}$, $\vct{c}_i^\transpose=m_{i,j}^{-1}M_{(i)}$. Then, $M=\vct{r}_j\vct{c}_i^\transpose$.
\end{enumerate}
\end{fact}

\begin{theorem}\label{thm:rank1Pencil}
Let $A,B\in\Re^{m \times n}$ be such that $A,B\neq\vct{0}_{m\times n}$ and $\rank{A+B}>1$. Choose any $(i,j)\in\{1\dots m\}\times \{1\dots n\}$ such that $a_{i,j}\neq0$. Such an $a_{i,j}$ is guaranteed to exist since $A\neq\vct{0}_{m\times n}$. Construct $\vct{r}_j(\lambda)$ and $\vct{c}_i(\lambda)^\transpose$ as 
\begin{align}\label{eqn:rc}
    \vct{r}_j(\lambda)=A^{(j)}+\lambda B^{(j)}, \quad \vct{c}_i(\lambda)^\transpose=\frac{1}{a_{i,j}+\lambda b_{i,j}}(A_{(i)}+\lambda B_{(i)})
\end{align}
Then, there exists $\lambda^*\in\Ce$ such that $\rank{A+\lambda^* B}=1$ if and only if either:
\begin{enumerate}
    \item $b_{i,j}\neq0$ and $\rank{A+\frac{-a_{i,j}}{b_{i,j}} B}=1$; or \label{cond:thm:rank1Pencil:bij}
    \item $A+\lambda^*B=\vct{r}_j(\lambda^*)\vct{c}_i(\lambda^*)^\transpose$. \label{cond:thm:rank1Pencil:decomp}
\end{enumerate}
\end{theorem}
\begin{proof}
    We first prove the forward direction. Suppose there exists $\lambda^*\in\Ce$ such that $\rank{A+\lambda^* B}=1$. We split the proof of the forward direction into two cases:
    
    {\bf Case $\lambda^*=\frac{-a_{i,j}}{b_{i,j}}$:} Suppose that $\lambda^*=\frac{-a_{i,j}}{b_{i,j}}$. Since $\lambda^*\in\Ce$, this implies $b_{i,j}\neq0$. Furthermore, $\rank{A+\lambda^* B}=1$ and $\lambda^*=\frac{-a_{i,j}}{b_{i,j}}$ implies $\rank{A+\frac{-a_{i,j}}{b_{i,j}} B}=1$. In addition, note that $\lambda^*=\frac{-a_{i,j}}{b_{i,j}}$ implies that $\vct{c}_i(\lambda^*)$ is undefined; therefore, the expression $A+\lambda^*B=\vct{r}_j(\lambda^*)\vct{c}_i(\lambda^*)^\transpose$ is undefined and cannot hold true.
    
    {\bf Case $\lambda^*\neq \frac{-a_{i,j}}{b_{i,j}}$:} Now, suppose $\lambda^*\neq\frac{-a_{i,j}}{b_{i,j}}$. Then $\rank{A+\frac{-a_{i,j}}{b_{i,j}} B}\neq1$. Also, $a_{i,j}+\lambda^*b_{i,j}\neq0$, so $\vct{c}_i(\lambda^*)$ is well-defined. Then $A+\lambda^*B=\vct{r}_j(\lambda^*)\vct{c}_i(\lambda^*)^\transpose$ follows from Fact \ref{fact:rank1facts}.
    
    Conversely, suppose that $b_{i,j}\neq0$ and $\rank{A+\frac{-a_{i,j}}{b_{i,j}} B}=1$. Then $\lambda^*=\frac{-a_{i,j}}{b_{i,j}}\in\Re\subset\Ce$ and $\rank{A+\lambda^* B}=1$. Of course, as in above, since $\vct{c}_i(\lambda^*)^\transpose$ is undefined at $\lambda^*=\frac{-a_{i,j}}{b_{i,j}}$ by definition, we conclude that $A+\lambda^*B\neq \vct{r}_j(\lambda^*)\vct{c}_i(\lambda^*)^\transpose$.
    
    Now, suppose $A+\lambda^*B=\vct{r}_j(\lambda^*)\vct{c}_i(\lambda^*)^\transpose$, which implies that $\vct{c}_i(\lambda^*)$ is well-defined. This implies that $b_{i,j}=0$ and/or $\rank{A+\frac{-a_{i,j}}{b_{i,j}} B}\neq1$. Furthermore, since $\lambda^*$ is the solution to a system of $m\times n$ linear or quadratic equations with real coefficients, we have $\lambda^*\in\Ce$. Then,  $\rank{A+\lambda^* B}=1$ follows from Fact \ref{fact:rank1facts}.
    \qed
\end{proof}

It is trivial to determine whether or not $b_{i,j}\neq0$ and $\rank{A+\frac{-a_{i,j}}{b_{i,j}} B}=1$. Thus, we will assume in the sequel that $\rank{A+\frac{-a_{i,j}}{b_{i,j}} B}\neq1$. 
Let us now examine the matrix equality $A+\lambda B=\vct{r}_j(\lambda)\vct{c}_i(\lambda)^\transpose$, introduce some additional notation, and state some lemmas that allow us to determine whether or not there exists a finite $\lambda^*$ such that $\rank{A+\lambda^* B}=1$.

With $\vct{r}_j(\lambda)$ and $\vct{c}_i(\lambda)$ as defined in Theorem \ref{thm:rank1Pencil}, let us write the following system of equations: $A+\lambda B=\vct{r}_j(\lambda)\vct{c}_i(\lambda)^\transpose$, that is,
\begin{align}
    &\begin{bmatrix}\label{eq:rank1Pencil:system}
     a_{1,1}+\lambda b_{1,1} &  \dots & a_{1,n}+\lambda b_{1,n} \\
    \vdots                 & \ddots  & \vdots                          \\
    a_{m,1}+\lambda b_{m,1}  & \dots & a_{m,n}+\lambda b_{m,n}
    \end{bmatrix}\\
    &=
    \begin{bmatrix}\label{eq:rank1Pencil:decompMatrix}
    \frac{(a_{1,j}+\lambda b_{1,j})(a_{i,1}+\lambda b_{i,1})}{a_{i,j}+\lambda b_{i,j}} & \dots & \frac{(a_{1,j}+\lambda b_{1,j})(a_{i,n}+\lambda b_{i,n})}{a_{i,j}+\lambda b_{i,j}} \\
    \vdots                           & \ddots& \vdots                          \\
    \frac{(a_{m,j}+\lambda b_{m,j})(a_{i,1}+\lambda b_{i,1})}{a_{i,j}+\lambda b_{i,j}} &  \dots & \frac{(a_{m,j}+\lambda b_{m,j})(a_{i,n}+\lambda b_{i,n})}{a_{i,j}+\lambda b_{i,j}}
    \end{bmatrix}
\end{align}
Since $\lambda$ is a scalar variable, it is clear from \eqref{eq:rank1Pencil:decompMatrix} that \eqref{eq:rank1Pencil:system} only has a solution (or possibly multiple solutions), $\lambda^*$, if $\lambda^*$ simultaneously satisfies $m\times n$ single-variable polynomials, where each polynomial is of degree at most $2$. Thus, one could solve all $m\times n$ single-variable polynomials and then check whether or not every solution has a common value. While this procedure is somewhat efficient, we'll show below that at most $(m-1)\times(n-1)$ of the polynomials have finite solutions and therefore contribute any meaningful information.  In addition, we'll show that it is sufficient to identify one polynomial that is not the zero polynomial and then conduct a matrix checking problem for the solution(s) of that polynomial.

Let us now introduce notation for the polynomial represented by row $s$ and column $t$ in \eqref{eq:rank1Pencil:decompMatrix}. For $(s,t)\in\{1\dots m\}\times \{1\dots n\}$, let 
\begin{equation}\label{eq:fst}
f_{s,t}(i,j;\lambda)=a_{s,t}+\lambda b_{s,t}-\frac{(a_{s,j}+\lambda b_{s,j})(a_{i,t}+\lambda b_{i,t})}{a_{i,j}+\lambda b_{i,j}}
\end{equation}

From \eqref{eq:fst}, it is clear that when $s=i$ or $t=j$, then $f_{i,t}(i,j;\lambda)$ and $f_{s,j}(i,j;\lambda)$ are the zero polynomial.  In other words, $f_{s,t}(i,j;\hat{\lambda})=0$ trivially holds true for all $\hat{\lambda}\in\Ce$ for one entire column and one entire row of \eqref{eq:rank1Pencil:decompMatrix}. Since these $m+n-1$ expressions hold true for all $\hat{\lambda}\in\Ce$, they lend no information for determining whether or not there exists $\lambda^*$ such that $\rank{A+\lambda^* B}=1$. Thus, we can disregard these $m+n-1$ polynomials and only consider the remaining $(m-1)\times(n-1)$ polynomials. We show that at least one of the remaining $(m-1)\times(n-1)$ polynomials is not the zero polynomial and present a method for determining whether or not there exists a $\gamma^*$ such that the pencil $A+\gamma^*B$ is a rank-1 pencil.    



\begin{theorem}\label{thm:rank1PencilSolveLambda}
    Consider non-zero matrices $A,B\in\Re^{m \times n}$ with $\rank{A+B}>1$. Pick any $(l,k)\in\{1\dots m\}\times \{1\dots n\}$ such that $f_{l,k}(i,j;\lambda)$ is not the zero polynomial. Let $\hat{\lambda}_1,\hat{\lambda}_2$ be solutions to $f_{l,k}(i,j;\lambda)=0$. Then, there exists $\lambda^*\in\Ce$ such that $\rank{A+\lambda^*B}=1$ if and only if $A+\hat{\lambda}_1 B=\vct{r}_j(\hat{\lambda}_1)\vct{c}_i(\hat{\lambda}_1)^\transpose$ or/and $A+\hat{\lambda}_2 B=\vct{r}_j(\hat{\lambda}_2)\vct{c}_i(\hat{\lambda}_2)^\transpose$, where $\vct{r}_j(\lambda)$ and $\vct{c}_i(\lambda)$ are defined in \eqref{eqn:rc}. 
\end{theorem}
\begin{proof}
We divide the proof into three steps:

    {\it Step 1:} Let us first prove that since $\rank{A+B}>1$, there exists at least one pair $(l,k)\in\{1\dots m\}\times \{1\dots n\}$ such that $f_{l,k}(i,j;\lambda)$ is not the zero polynomial. Suppose, by way of contradiction, that $f_{s,t}(i,j;\lambda)$ is the zero polynomial for all $(s,t)\in\{1\dots m\}\times \{1\dots n\}$. This implies that for any $\hat{\lambda}\in\Ce$ and for all $(s,t)\in\{1\dots m\}\times \{1\dots n\}$, $f_{s,t}(i,j;\hat{\lambda})=0$. Furthermore, this implies that for any $\hat{\lambda}\in\Ce$,  $A+\hat{\lambda} B=\vct{r}_j(\hat{\lambda})\vct{c}_i(\hat{\lambda})^\transpose$ and $\rank{A+\hat{\lambda} B}=1$.  In particular, $\rank{A+B}=1$, which is a contradiction. \
    
    {\it Step 2:} Now, suppose there exists $\lambda^*\in\Ce$ such that $\rank{A+\lambda^* B}=1$. Then, by Theorem \ref{thm:rank1Pencil}, $A+\lambda^*B=\vct{r}_j(\lambda^*)\vct{c}_i(\lambda^*)^\transpose$. Thus, $f_{s,t}(i,j;\lambda^*)=0$ for all $(s,t)\in\{1\dots m\}\times \{1\dots n\}$. In particular, $f_{l,k}(i,j;\lambda^*)=0$. Therefore, either $\hat{\lambda}_1=\lambda^*$ or/and $\hat{\lambda}_2=\lambda^*$. It then follows that $A+\hat{\lambda}_1B=\vct{r}_j(\hat{\lambda}_1)\vct{c}_i(\hat{\lambda}_1)^\transpose$ or/and $A+\hat{\lambda}_2 B=\vct{r}_j(\hat{\lambda}_2)\vct{c}_i(\hat{\lambda}_2)^\transpose$.
    
    {\it Step 3:} In the other direction, let $\hat{\lambda}_1,\hat{\lambda}_2$ be solutions to $f_{l,k}(i,j;\lambda)=0$. Note that $f_{l,k}(i,j;\lambda)=0$ may be a linear equation.  In that case, for simplicity, let $\hat{\lambda}_1=\hat{\lambda}_2$. 
    Now, suppose $A+\hat{\lambda}_1B=\vct{r}_j(\hat{\lambda}_1)\vct{c}_i(\hat{\lambda}_1)^\transpose$ and let $\lambda^*=\lambda_1$. Then $A+\lambda^*B=\vct{r}_j(\lambda^*)\vct{c}_i(\lambda^*)^\transpose$ and $\rank{A+\lambda^*B}=1$ by Theorem \ref{thm:rank1Pencil}.  The case of $A+\hat{\lambda}_2 B=\vct{r}_j(\hat{\lambda}_2)\vct{c}_i(\hat{\lambda}_2)^\transpose$ is similar and therefore omitted.  
    \qed
\end{proof}

Let us use $\Lambda^1:\Rea^{m\times n}\times\Rea^{m\times n}\rightarrow 2^{\Ce}$ to represent the solution set obtained from Theorem \ref{thm:rank1PencilSolveLambda}, that is, 
\begin{align}
    \Lambda^1(A,B) = \Big\{\lambda^*\in \Ce: \rank{A+\lambda^*B} = 1\Big\}.\label{eqn:LambdaAB}
\end{align}
Note that $\Lambda^1(A,B)$ has a maximum cardinality of $2$ and may be empty. Further, Theorem \ref{thm:rank1Pencil} implies that if $b_{i,j}\neq0$ and $\rank{A+\frac{-a_{i,j}}{b_{i,j}} B}=1$, then $\frac{-a_{i,j}}{b_{i,j}}\in\Lambda^1(A,B)$. 

Similar to Subsection \ref{subsec:matrixPencils}, we are only concerned with the real, strictly positive eigenvalues. Therefore, we define the restricted set of eigenvalues as:
\begin{equation*}
    \Lambda^1_{>0}(A,B)=\{\lambda\in \Re_{>0}\vert \lambda\in\Lambda^1(A,B)\} = \Lambda^1(A,B)\cap (0,\infty).
\end{equation*}
Unlike $\Lambda_{>0}(A,B)$ defined in Subsection \ref{subsec:matrixPencils}, any $\lambda^*\in\Lambda^1_{>0}(A,B)$ guarantees that $\rank{A+\lambda^*B} = 1$. Thus, our modified three step process for determining whether there exists $(\hat{A},\hat{B})\in \Upsilon(\tilde A,\tilde B)$ such that $\rank{\hat{A}+\hat{B}}=1$ is:
\begin{enumerate}
\item Compute $(A^\dagger,\vct u^\dagger_A,\vct v^\dagger_A) = \Xi(\tilde A)$ and $(B^\dagger,\vct u^\dagger_B,\vct v^\dagger_B) = \Xi(\tilde B)$.
\item Compute $\Lambda^1_{>0}(A,B)$. 
\begin{enumerate}
    \item If $\Lambda^1_{>0}(A,B)\neq \emptyset$, let $\gamma^*=\lambda^*\in\Lambda^1_{>0}(A,B)$. Define $\hat{\vct u}$ and $\hat{\vct v}$ as
        \begin{align*}
         \hat{\vct u}:=\vct u_A^\dagger+\gamma^*\vct u^\dagger_B 
         \quad \hat{\vct v}:=\vct v_A^\dagger+\gamma^*\vct v^\dagger_B.
        \end{align*}
        \item Output $\Pi^\dagger\Big(\tilde A,\tilde B\Big) := \Big(\hat A,\hat B\Big)$, where $\hat A = \tilde A - \vct{1}_m\hat{\vct{u}}^\transpose$ and $\hat B = \gamma^*\tilde B -\hat{\vct{v}}\vct{1}_n^\transpose$.
\end{enumerate}
\item If $\Lambda^1_{>0}(A,B)= \emptyset$, output ``No such rank-1 game exists"
\end{enumerate}

Unlike the three step process that we developed in Subsection \ref{sub:MainResult}, the output of the process here is either a strategically equivalent rank-1 game or the fact that no such game exists. In the next section we discuss generic games and algorithmic implications.

\section{Some Results on Rank Reduction in Generic Games}\label{sec:generic}
In this section, we collect some natural consequences of Proposition \ref{prop:colSpanC2} for reducing the rank of generic games. Recall that a game is determined by the matrices $(\tilde A,\tilde B)\in\Re^{2\times(m\times n)}$. Thus, one can view $\Re^{2\times(m\times n)}$ as the space of all games. If one picks a generic game from this space, it is natural to ask whether or not the rank of the game can be reduced and by how much. 

It turns out that Proposition \ref{prop:colSpanC2} allows us to conclude the following two results, where we only focus on the case of $\gamma = 1$ to ease analysis. In the first result, we focus on generic games in which $m=n$. It should be noted that the set of generic games in the space $\Re^{2\times(m\times m)}$ with full rank has full measure. We conclude that for such games, we can reduce the rank only by 1. In the second result, we consider the case where $m<n$.

\begin{proposition}\label{prop:squaregame}
If $m=n$, then for almost every (random) game $(m,m,A,B)$, the rank of the game can only be reduced by 1.
\end{proposition}
\begin{proof}
    In order to prove the result, we show that the set of games whose rank can be reduced by two has Lebesgue measure 0. We prove this in two steps:
    
    \noindent {\it Step 1:} First, recall that for almost every game $(m,m,A,B)$, $rank(C) = m$. Note that by Proposition \ref{prop:colSpanC2}, we can reduce the rank of the game by 2 if there exists $\vct{x}_1$ satisfying $\vct{1}_m^\transpose \vct{x}_1 = 0$ and $C\vct{x}_1 = \vct{1}_m$. Since $C$ is full rank, there exists a unique $\vct{x}_1$ satisfying $C\vct{x}_1 = \vct{1}_m$, which is given by $\vct{x}_1 = C^{-1}\vct{1}_m$. Thus, the rank of the game can be reduced by 2 if and only if $\vct{1}_m^\transpose C^{-1}\vct{1}_m = 0$. This holds if the cofactors of the matrix $C$ sum to zero. The sum of cofactors is a multivariate polynomial of degree $m-1$, and thus, the set of points where this polynomial equals 0 is an $(m-1)$ dimensional manifold in a $m^2$ dimensional space. Consequently, the set of all $C\in\Rea^{m\times m}$ whose cofactors sum to zero has Lebesgue measure 0. 
    
    \noindent {\it Step 2:} Now, we note that $B = C-A$. Since for almost every $C$, the sum of cofactors is not zero, we conclude that the for almost every game $(m,m,A,B)$, the rank can only be reduced by 1. The proof of the result is complete.
    \qed 
\end{proof}


In contrast to the case considered above, we next consider games with $m\leq n$, rank of the sum $C$ of payoff matrices is $k\leq m$, $\vct{1}_m\in\colspan{C}$, and $\vct{1}_n\in\colspan{C^\transpose}$ is satisfied. Let us define this space as $\mathcal C_k$:
\begin{equation}
    \mathcal C_k = \Big\{C\in\Re^{m\times n} \big| rank(C) = k, \vct{1}_m\in\colspan{C}, \vct{1}_n\in\colspan{C^\transpose}\Big\}
\end{equation}
\begin{proposition}
Let $2\leq m\leq n$. Consider a rank-$k$ game $(m,n,A,B)$, and assume that $C:= A+B \in\mathcal C_k$. Split the matrix $C$ as $C =  \begin{bmatrix} C_{11} & C_{12}\\ C_{21} & C_{22}  \end{bmatrix}$, where $C_{11}$ is a $k\times k$ full rank submatrix. Then, the rank of the game can be reduced by 2 if and only if there exists a vector $\vct x_{12}\in\Re^{n-k}$ such that
\begin{align}\label{eqn:Creduce2}
 \mat{C_{22} - C_{21} C_{11}^{-1}C_{12}\\\vct 1_{n-k}^\transpose - \vct{1}_k^\transpose C_{11}^{-1}C_{12}}\vct x_{12} = 
 \mat{\vct{1}_{n-k} - C_{21}C_{11}^{-1} \vct 1_k \\-\vct{1}_k^\transpose C_{11}^{-1}\vct{1}_k }.
\end{align}
\end{proposition}
\begin{proof}
    Let $\vct{x}_1$ be such that $C\vct{x}_1 = \vct{1}_m$. Split the matrix $C$ and $\vct{x}_1$ as $\begin{bmatrix} \vct{x}_{11}\\ \vct{x}_{12}  \end{bmatrix}$, where $\vct{x}_{11}$ is a $k\times 1$ vector, and $\vct{x}_{12}$ is a $(n-k)\times 1$ vector. We also note here that if $k = m$, then $C_{21}$ and $C_{22}$ are empty matrices. Now, due to Proposition \ref{prop:colSpanC2} in Appendix \ref{app:Drankred}, we can reduce the rank of the game by 2 if there exists $\vct{x}_1$ satisfying $\vct{1}_n^\transpose \vct{x}_1 = 0$ and $C\vct{x}_1 = \vct{1}_m$. Thus, $C\vct{x}_1 = \vct{1}_m$ implies
    \begin{align*}
        C_{11}\vct{x}_{11} + C_{12} \vct{x}_{12} &= \vct{1}_k \implies
        \vct{x}_{11}  = C_{11}^{-1}\vct{1}_k - C_{11}^{-1}C_{12} \vct{x}_{12},\\
        C_{21}\vct{x}_{11} + C_{22} \vct{x}_{12} &= \vct{1}_{n-k}, \\
        \implies
        C_{22}\vct{x}_{12}  & = \vct{1}_{n-k} - C_{21}C_{11}^{-1} \vct 1_k + C_{21} C_{11}^{-1}C_{12} \vct{x}_{12},\\
        \implies \Big[C_{22} - C_{21} C_{11}^{-1}C_{12}\Big]\vct  x_{12} & = \Big(\vct{1}_{n-k} - C_{21}C_{11}^{-1} \vct 1_k \Big).
    \end{align*}
    Now, substituting the above expressions in the equation $\vct{1}_n^\transpose \vct{x}_1 = 0$ yields
    \begin{align*}
        \vct{1}_k^\transpose\vct{x}_{11} + \vct{1}_{n-k}^\transpose\vct{x}_{12} & =  \vct{1}_k^\transpose C_{11}^{-1}\vct{1}_k +\Big(\vct 1_{n-k}^\transpose - \vct{1}_k^\transpose C_{11}^{-1}C_{12}\Big)\vct x_{12} = 0\\
        \implies \Big(\vct 1_{n-k}^\transpose - \vct{1}_k^\transpose C_{11}^{-1}C_{12}\Big)\vct x_{12}  &= -\vct{1}_k^\transpose C_{11}^{-1}\vct{1}_k 
    \end{align*}
    which yields the result.\qed
\end{proof}

In \cite[Lemma~2.5]{theobald2009enumerating}, the author shows that all $m\times m$ games of rank-$m$ are strategically equivalent to a game of rank $m-1$. We briefly restate that lemma here in order to compare and contrast our results.

\begin{lemma}[{\citet[Lemma~2.5]{theobald2009enumerating}}]\label{lem:theobaldsLemma}
Let $(m,m,A,B)$ be a game of rank $m$. Then there exists a game of rank ($m-1$) with the same set of Nash equilibria.
\end{lemma}
\begin{lemma}\label{lem:thm:generalGame1m:fullRank}
    Consider the game  $(m,n,\bar{A},\bar{B})$ with $n\geq m\geq2$, $\bar{C}=\bar{A}+\bar{B}$, and $\rank{\bar{C}}=\bar{k}=m$. The game $(m,n,\bar{A},\bar{B})$ is strategically equivalent to a game of rank-$(\bar{k}-1)$
\end{lemma}
\begin{proof}
   Clearly $\rank{\bar{C}}=m$ implies that there exists $\vct x_1\in\Rea^n$ such that $\bar{C}\vct{x}_1=\vct{1}_m$. The result then directly follows from Theorem \ref{thm:theMainResult}.\qed
\end{proof}

Our result in Lemma \ref{lem:thm:generalGame1m:fullRank} generalizes Lemma \ref{lem:theobaldsLemma} to the game $(m,n,A,B)$, with $m\leq n$ and $\rank{A+B}=m$. In addition, Theorem \ref{thm:theMainResult} further generalizes the result to hold for certain games that are not full rank.  Furthermore, for all cases, our results are constructive and provide a method for calculating the lower-rank equivalent game.

\section{A Numerical Example and Algorithmic Implications}\label{sec:numerical}

In this section, we show a small numerical example and discuss the complexity of our proposed algorithms.
\subsection{Numerical Example}\label{subsec:numerical}
Let us consider the following bimatrix game: In this game, $m=n=6$ and is a full rank game.
\begin{align*}
 &\tilde A = 
 \mat{-32 & -98 & 39  & 12   & -66  & -67\\
       0  & 74  & -29 & -28  &  62  &   21\\
    56  & -14  & -33  &  20  & -42  &  -15\\
  -109  &-121  & -38  & -31  & -57  &  -92\\
     6  &  42  & -63  & -26  &  20  &   1\\
   -40  & -84  &  11  &  -2  & -44  & -59
   }
 \tilde B = \mat{ 
    -32  & -98 &   39 &   12 &  -66 &  -67\\
     0   & 74  & -29  & -28   & 62  &  21\\
    56   &-14  & -33  &  20   &-42 &  -15\\
  -109  &-121   &-38  & -31  & -57 &  -92\\
     6   & 42   &-63   &-26  &  20  &   1\\
   -40  & -84  &  11  &  -2  & -44  & -59
   }\\
   & \text{Singular values of } \tilde C = \{210.53, 66.25, 59.43, 14.65, 3.44, 1.49\}
\end{align*}
In the Step 2 of the rank reduction process, we get $\Lambda(\bar A,\bar B) = \{2,2,4,4\}$, so we pick $\gamma^* = 4$. After applying Step 4 of the rank reduction process, we arrive at the following game, which is strategically equivalent to the game above.
\begin{align*}
 & \hat A = \mat{
    -8  &-104 &   60  &  28  & -84 &  -72\\
    24  &  68  &  -8 &  -12 &   44  &  16\\
    80  & -20  & -12 &   36  & -60  & -20\\
   -85  &-127 &  -17 &  -15 &  -75  & -97\\
    30  &  36  & -42 &  -10 &    2   & -4\\
   -16   &-90  &  32  &  14 &  -62 &  -64
   }
\hat B = \mat{
     8 &  104 &  -60 &  -28 &   84   & 72\\
   -24  &-136  &   8   & -4  &-132 &  -68\\
   -80  &  36 &   12  & -28  &  68 &   28\\
    85  &  49  &  17   & -3  & -27 &   37\\
   -30  & -90  &  42   & -2  & -74 &  -38\\
    16   & 68  & -32  & -16  &  24 &   44\\
} \\
& \text{Singular values of } \hat C = \{219.97, 11.96\}.
\end{align*}
As we see, the lower rank game is a game of rank 2. It is also interesting to note that if we instead pick $\gamma^* = 2$ (since it also has multiplicity 2), and we arrive at a different rank-2 game after the rank reduction process is complete.

\subsection{Algorithmic Implications}\label{subsec:algorithm}

Since our algorithms are algebraic in nature, we use here the common arithmetic computational cost model where one considers the number of arithmetic operations and comparisons \citep{golub2012matrix,pan1999complexity} as simply `operations'. We also, without loss of generality, assume $m\leq n$. If this is not the case, one can simply consider the transposed game, $(n,m,B^\transpose,A^\transpose)$, and all results in this section still hold.

First, considering the map, $\Psi(D)$, checking conditions 1,2, and 3 is equivalent to determining the rank of a properly augmented matrix.  For example, to check condition 2 one could verify: 
    \begin{equation*}
      rank\left(\begin{bmatrix} D & \vct{1}_m  \end{bmatrix}\right) = rank(D).
    \end{equation*}
Condition 4 is simply the solution of a linear system. Both solving for rank and solving a linear system are efficiently computable via a variety of methods, and can be done in at most $\bigO(n^3)$ operations. Computing $\vct{u}$ and $\vct{v}$ for the map $\Psi(D)$ is simply selecting a row (column) of a matrix, a scalar-vector multiplication, and vector addition, which can all be done in $\bigO(n)$ operations. Similarly, the map $\Xi(D)$ requires matrix-vector multiplication, matrix multiplication, and matrix addition, thus can be completed in $\bigO(n^3)$ operations.

Calculating $\gamma^*:=\gamma^*(A^\dagger,B^\dagger)$ requires first calculating the TWCF in \eqref{eq:lem:TW1} using elementary row and column operations. This is equivalent to pre- or post-multiplying $A^\dagger,B^\dagger$ by square, nonsingular matrices and requires $4(2m^2n +2mn^2)$ operations. If necessary, calculating \eqref{eq:lem:TW2} is similar, and also requires $4(2m^2n +2mn^2)$ operations. So, each iteration requires $\bigO(n^3)$ operations. At each iteration of Algorithm \ref{alg:TWCF}, the matrices $A_i,B_i$ have dimensions that are strictly smaller than the original dimensions ($m_i<m,n_i<n$); however, the decrease in dimensionality per iteration is problem dependent. While the algorithm can terminate earlier, it must eventually terminate when $A_i,B_i\in\Rea^{1\times 1}$, or after at most $m$ iterations. Thus, at the worst case, with $n=m$, Algorithm \ref{alg:TWCF} requires $\bigO(n^4)$ operations. This still leaves the problem of computing the eigenvalues, of $(E_{11,i},I_{r_i})$, but this only has to be done once, on matrices with dimensions $m_i=n_i\leq n$. Since eigenvalues can be calculated, up to fixed precision, $\epsilon$, in $\bigO(n^3)$ \citep[p. 391]{golub2012matrix}, \citep{ pan1999complexity}, the overall running time is dominated by the running time of $\bigO(n^4)$.

Finally, the map $\Pi\Big(A, B\Big)$ requires vector-vector multiplication and matrix addition, and can be done in $\bigO(mn)$ operations. Overall, the running time of the algorithm is, not surprisingly, dominated by solving the matrix pencil problem and has complexity $\bigO(n^4)$.

In comparison, for the special case of rank-1 matrix pencils, the algorithm presented in Section \ref{sec:rank1Pencils} runs in $\bigO(mn)$ operations. To see this, we first note that by Fact \ref{fact:rank1facts}, determining whether a matrix is rank-$1$ or not is equivalent to $mn$ divisions and $mn$ comparisons, and therefore $\bigO{(mn)}$. As mentioned in Section \ref{sec:rank1Pencils}, finding an $(s,t)\in\{1\dots m\} \times \{1\dots n\}$ such that $f_{s,t}(i,j;\lambda)$ is not the zero polynomial requires a search over a space of $(m-1)\times(n-1)$, thus time $\bigO(mn)$. Calculating the coefficients of $f_{s,t}(i,j;\lambda)$ requires at most $4$ scalar multiplications and $3$ scalar additions/subtractions, which requires time $\bigO(1)$. Now, at worst solving $f_{s,t}(i,j;\lambda)=0$ is equivalent to computing a square root, which has the same complexity as multiplication \citep{alt1979square}. Furthermore, we note that only one such $f_{s,t}(i,j;\lambda)=0$ must be solved, thus this step is independent of the size of the game. With candidate values of $\hat{\lambda}$ thus determined, checking whether $\bar{A}+\hat{\lambda}\bar{B}=\vec{r}_j(\hat{\lambda})\vec{c}_i(\hat{\lambda})^\transpose$ requires one vector outer product and one matrix comparison, thus takes time $\bigO(mn)$. Therefore, for this special case, we remove the matrix pencil bottleneck and the overall running time of our algorithm is dictated by the maps $\Psi(D)$ and $\Xi(D)$, with complexity $\bigO(n^3)$.

\section{Conclusion}\label{sec:conclusion}
Nonzero-sum games have been shown to be computationally challenging to solve. In this paper, we present an alternative approach to computing equilibrium in a certain class of nonzero-sum games. Given a nonzero-sum game, our approach exploits the strategic equivalence between bimatrix games to construct a lower rank bimatrix game that is strategically equivalent to the original game via a positive affine transformation. Moreover, our technique is constructive, that is, we present an algorithm to reduce the rank of the given game and this algorithm is efficient (the runtime complexity is dominated by the execution runtime of the matrix pencil problem).

Our approach has the potential to reduce the rank of the game substantially in some cases. If the original game can be reduced to a rank-0 (also known as zero-sum game) or a rank-1 game, then we know that it can be solved efficiently using algorithms involving linear programs or parameterized linear programs, respectively. If the game is of low rank or sparse, one can employ polynomial time approximation algorithms to determine an approximate Nash equilibrium. Thus, our rank reduction approach substantially expands the class of games that can be solved in polynomial time.

One important problem left open for the future is as follows. Recall that the weakest definition of strategic equivalence was introduced in \cite{liu1996invariance}, in which two games are strategically equivalent if and only if their best response correspondences are equal. Can we devise an algorithm that outputs a lower rank game that preserves the best response correspondence? Such an approach can further expand the class of games that are solvable in polynomial time. We leave this problem for future research. 

\appendix

\section{Proof of Lemma \ref{lem:Drankred}}\label{app:Drankred}

The proof of Lemma \ref{lem:Drankred} is divided into four cases.

{\bf Case 1:} For this case, we cannot apply the Wedderburn rank reduction formula to yield a strategically equivalent game. Thus, $\hat{\vct u} = \vct 0_n$, $\hat{\vct v} = \vct 0_m$, and $l = 0$. 
{\bf Case 2:} If $\vct{1}_m\in\colspan{D}$, then there exists $\vct x_1\in\Rea^n$ and $\vct y_1  = \vct e_i \in\Rea^m$ such that	$D\vct{x}_1=\vct{1}_m$ and $w_1=\vct{y}_1^\transpose D\vct{x}_1 = 1$. Let $\hat{\vct{u}}^\transpose=w_1^{-1}\vct{y}_1^\transpose D = D_{(i)}^\transpose$ and compute $D_2$ using \eqref{eq:Wedderburn} as follows:
\begin{equation*}
    D_2 = D-w_1^{-1}D\vct{x}_1\vct{y}_1^\transpose D=D-\vct{1}_m\hat{\vct{u}}^\transpose.
\end{equation*}

{\bf Case 3:} If $\vct{1}_n\in\colspan{D^\transpose}$, then we apply Case 2's operation to $D^\transpose$ by taking $\vct x_1 = \vct e_j\in\Re^n$ and defining $\hat{\vct v} = D_{(j)}$.

Before we address Case 4, we first state a proposition that is necessary for the proof. Motivated by \cite[Theorem~2.1]{chu1995rank}, we have the following proposition that presents conditions under which a chosen vector is in the row span of the matrix $C_2$ obtained after one application of \eqref{eq:Wedderburn}.
\begin{proposition}\label{prop:colSpanC2}
We use here the same notation as in Theorem \ref{thm:wedderburn}. Let $\rank{C}=k$, with $k \geq 2 $. Let $\{\vct{x}_1,\vct{y}_1\}$ be vectors associated with a rank-reducing process (so that $\vct{y}_1^\transpose C \vct{x}_1\neq0$). We have $\vct{z}\in \colspan{C_2^\transpose}$ if and only if $\vct{z}\in \colspan{C^\transpose}$ and $\vct{z}\perp\vct{x}_1$.
\end{proposition}
\begin{proof}

Suppose that $\vct{z}\in \colspan{C_2^\transpose}$. Then there exists a $\vct{y}_2\in \Re^m$ such that $C_2^\transpose \vct{y}_2=\vct{z}$. Choose such a $\vct{y}_2$ and define $\vct{v}_2$ as
\begin{equation}\label{eq:v2Def}
    \vct{v}_2\coloneqq \vct{y}_2-
	\frac{\vct{y}_2^\transpose C \vct{x}_1}{\vct{y}_1^\transpose C \vct{x}_1}\vct{y}_1. 
\end{equation}

Directly from Theorem \ref{thm:wedderburn}, we have
	\begin{align}
	\vct{y}_2^\transpose C_2&= \vct{y}_2^\transpose C-w_1^{-1}\vct{y}_2^\transpose C\vct{x}_1\vct{y}_1^\transpose C
    =\bigg(\vct{y}_2^\transpose- \frac{\vct{y}_2^\transpose C \vct{x}_1}{\vct{y}_1^\transpose C \vct{x}_1}\vct{y}_1^\transpose \bigg)C
    = \vct{v}_2^\transpose C \label{eq:prop:colSpanC2:v2C}
	\end{align}
Then by \eqref{eq:prop:colSpanC2:v2C}, $\vct{z}=C^\transpose\vct{v}_2$ which implies that $\vct{z}\in \colspan{C^\transpose}$. From the proof of Theorem \ref{thm:wedderburn}, it is clear that $\vct{x}_1\in \nullspace{C_2}$. Then
\begin{equation*}
   \vct{z}^\transpose \vct{x}_1 = \vct{y}_2^\transpose C_2\vct{x}_1=\vct{v}_2^\transpose C \vct{x}_1=0,
\end{equation*}
which implies that $\vct{z}\perp\vct{x}_1$.

Now, suppose that $\vct{z}\in \colspan{C^\transpose}$ and $\vct{z}\perp\vct{x}_1$. Since $\rank{C}\geq2$, the $\colspan{C^\transpose}$ is at least a two dimensional subspace, and thus there exists a $\vct{y}_1\in \Re^m$ such that $C^\transpose \vct{y}_1$ and $\vct{z}$ are linearly independent. Choose such a $\vct{y}_1$. Then, $\vct{y}_1^\transpose C\vct{x}_1\neq0$. Pick $\vct{y}_2$ such that $C^\transpose \vct{y}_2=\vct{z}$ and compute $\vct{v}_2$ as in \eqref{eq:v2Def}. Then,
\begin{equation}\label{eq:prop:colSpanC2:v2Equaly2}
    \vct{v}_2=\vct{y}_2-
	\frac{\vct{z}^\transpose\vct{x}_1}{\vct{y}_1^\transpose C \vct{x}_1}\vct{y}_1
	=\vct{y}_2-\frac{0}{\vct{y}_1^\transpose C \vct{x}_1}\vct{y}_1=\vct{y}_2.
\end{equation}
Due to \eqref{eq:prop:colSpanC2:v2C}, we have $\vct{y}_2^\transpose C_2=\vct{v}_2^\transpose C$, which implies $\vct{y}_2^\transpose C=\vct{z}$ from  \eqref{eq:prop:colSpanC2:v2Equaly2}. Thus, $\vct{z}\in \colspan{C_2^\transpose}$. \qed
\end{proof}

The requirement of $\vct{z}\in \colspan{C^\transpose}$ and $\vct{z}\perp\vct{x}_1$ can be equivalently written as
\begin{align}
 \begin{bmatrix} C & C\vct x_1\\ \vct z^\transpose & 0  \end{bmatrix}\mat{\vct x_1\\ -1} = \mat{0\\0}.
\end{align}
 
{\bf Case 4(a):} If $\vct 1_m \in\colspan{D}$, $\vct 1_n \in\colspan{D^\transpose}$, but the nullspace of $\begin{bmatrix} D & \vct{1}_m\\ \vct{1}^T_n & 0  \end{bmatrix}$ is of the form $\mat{\vct x\\ 0}$, then Proposition \ref{prop:colSpanC2} implies that $\vct 1_n\not\in\colspan{D_2}$. Thus, the rank can only be reduced by $l = 1$. Using the same operation as in Case 2, we get $\hat{\vct u} = D_{(i)}$ for any $i\in\{1,\ldots,m\}$, $\hat{\vct v} = \vct 0_m$.

{\bf Case 4(b):}
If $\vct 1_m \in\colspan{D}$, $\vct 1_n \in\colspan{D^\transpose}$, and the nullspace of $\begin{bmatrix} D & \vct{1}_m\\ \vct{1}^T_n & 0  \end{bmatrix}$ contains a vector of the form $\mat{\vct x\\ -1}$, then by Proposition \ref{prop:colSpanC2}, there exists $\vct{x}_1,\vct{x}_2 = \vct e_j\in\Rea^n$ and $\vct{y}_1 = \vct e_i,\vct{y}_2\in\Rea^m$ such that:
\begin{enumerate}
    \item $w_1=\vct{y}_1^\transpose D\vct{x}_1\neq 0$. Let $D\vct{x}_1=\vct{1}_m$, $\hat{\vct{u}}^\transpose=w_1^{-1}\vct{y}_1^\transpose D$, and compute $D_2$ using \eqref{eq:Wedderburn} as follows:
\begin{equation*}
    D_2 = D-w_1^{-1}D\vct{x}_1\vct{y}_1^\transpose D=D-\vct{1}_m\hat{\vct{u}}^\transpose. 
\end{equation*}
\item $\vct{y}_2^\transpose D_2=\vct{1}_n^\transpose$ and $w_2=\vct{y}_2^\transpose D_2\vct{x}_2\neq 0$. Let $\hat{\vct{v}}=w_2^{-1}D_2\vct{x}_2$ and compute $D_3$ using \eqref{eq:Wedderburn} as follows:
\begin{equation*}
    D_3 = D_2-w_2^{-1}D_2\vct{x}_2\vct{y}_2^\transpose D_2=D_2-\hat{\vct{v}}\vct{1}_n^\transpose=D-\vct{1}_m\hat{\vct{u}}^\transpose-\hat{\vct{v}}\vct{1}_n^\transpose. 
\end{equation*}
\end{enumerate}

\bibliographystyle{spbasic}      
\bibliography{ser1_bib}   

\end{document}